\documentclass{IEEEtran}

\usepackage{amsmath}
\usepackage{amssymb}
\usepackage{booktabs}
\usepackage{float}
\usepackage{cite}
\usepackage{amsthm}
\usepackage{enumitem}
\usepackage{algorithm}
\usepackage{algpseudocode}
\usepackage{todonotes}
\usepackage[colorlinks=true, urlcolor=blue, linkcolor=black, citecolor=black]{hyperref}
\newtheorem{lemma}{Lemma}
\newtheorem{proposition}{Proposition}

\newtheorem{definition}{Definition}
\usepackage{pgfplots}
\pgfplotsset{compat=1.18}
\usepgfplotslibrary{groupplots}
\usepackage{tikz}
\usepackage{bm}
\usepackage{verbatim}

\definecolor{customCyan}{RGB}{0, 150, 255}
\definecolor{customOrange}{RGB}{230, 100, 50}

\DeclareMathOperator*{\minimize}{minimize}
\DeclareMathOperator*{\maximize}{maximize}

\DeclareMathOperator*{\subjectto}{subject\ to}
\DeclareMathOperator{\diag}{diag}

\newcommand{\Epotential}{F}
\newcommand{\Esnapshot}{{\mathcal{E}}}
\newcommand{\Eviable}{{\mathcal{E}}'}

\author{Arman Mollakhani, Jerayu Tiamraj, Shu-Jie Cao, and Dongning Guo\thanks{The authors are with the Department of Electrical and Computer Engineering, Northwestern University, Evanston, IL 60208.}\thanks{Emails:
\{arman.mollakhani, shujie.cao, dguo\}@northwestern.edu,
tomtiamraj2025@u.northwestern.edu}\thanks{This work was supported in part by NSF grants Nos.~2132700 (SpectrumX) and 2434044.}
}

\title{Inter-Satellite Link Optimization for Low-Latency Global Networking}

\date{\today}

\begin{document}

\maketitle

\begin{abstract}
Large-scale low-Earth-orbit satellite constellations offer a promising platform for global low-latency networking, %
in part because signals propagate
faster %
in free space than in fiber and copper. In such systems, end-to-end latency is largely determined by the inter-satellite link (ISL) topology. In particular, the network diameter---the %
longest shortest path, in hops, between any pair of satellites---serves as a key performance metric for time-sensitive applications. Designing diameter-optimal topologies is challenging due to node degree constraints, orbital dynamics, and the combinatorial nature of link selection. %
This paper proposes a two-stage optimization framework for ISL topology design. First, a continuous relaxation of the link-selection problem is formulated as a convex program that maximizes the algebraic connectivity of the %
graph Laplacian, serving as a tractable surrogate for diameter minimization. Second, the resulting fractional solution is mapped to a feasible %
topology via maximum-weight b-matching, a problem solvable in polynomial time. %
Simulations on Walker–Delta constellations show that the proposed method consistently achieves smaller network diameters %
than a local-search baseline and standard patterned topologies.
For a %
1,584-satellite constellation in which each satellite has %
four ISLs of up to 2,500~km, the %
diameter %
is as small as %
12 hops, %
and the maximum %
propagation delay between any two satellites is about %
90~ms. With up- and downlinks adding up to 8~ms, the user-to-user propagation delay is at most
about 98~ms, %
comparable to the %
delay over a great-circle fiber route between antipodal points on Earth.
\end{abstract}

\section{Introduction}
\label{sec:introduction}

Low-Earth-orbit (LEO) constellations such as Starlink, AST SpaceMobile, OneWeb, and Amazon Leo
aim to provide %
near-ubiquitous global coverage~\cite{berry2024spectrum}. Optical inter-satellite links (ISLs) %
enable efficient in-network data transport within %
and %
across constellations.
In free space, electromagnetic waves propagate at approximately 
300~km per millisecond (ms), %
roughly 47\% faster than in optical fiber~\cite{handley2018delay}. Consequently, a well-designed LEO ISL network can outperform the fastest terrestrial fiber routes in end-to-end latency, making such
systems %
well-suited for latency-critical applications including %
remote control, telemedicine, high-frequency trading, and blockchain synchronization%
~\cite{bozkurt2017why, mollakhani2025fault}.

Constellations operated by a single service provider can further exploit globally optimized routing instead of distributed protocols such as the Border Gateway Protocol (BGP): forwarding tables can be precomputed for short \emph{epochs}, disseminated to all nodes, and updated as the topology evolves~\cite{handley2018delay}. This deterministic approach eliminates in-orbit route discovery and enables techniques such as cut-through switching~\cite{kermani1979virtual} or multiprotocol label switching%
~\cite{rfc3031}, %
reducing per-hop processing overhead to the microsecond scale~\cite{heine2023end}. %

In this work, we use the network diameter---the maximum %
hop distance between any two satellites---as a structural proxy for latency.
ISL lengths vary across links and over time, but
physical range constraints %
keep every link well below Earth's radius.
Reducing the maximum hop count therefore tightens a uniform upper bound on propagation delay and reduces relay processing. This work accounts solely for satellite-to-satellite propagation, leaving terrestrial routing out of scope.

In practice, %
each satellite is limited to a small number of steerable laser terminals, and any candidate link must satisfy distance and line-of-sight constraints not just at a single instant but ideally over an entire orbital period. The design problem is therefore to select a sparse, physically feasible set of %
ISLs that achieves a small diameter while maintaining orbital stability.
Most existing approaches use fixed geometric patterns, which cannot exploit the full set of feasible links, or local search, which is prone to poor local optima.
This paper addresses these limitations through a spectral graph-theoretic framework.
The main contributions of this paper are as follows:
\begin{enumerate}[label=(\roman*)]
    \item We formulate a constrained network diameter minimization problem %
    that accounts for hardware limits, geometric constraints, and long-term link stability. %

    \item We propose a two-stage spectral optimization framework that combines convex relaxation and integer linear programming to design 
    degree-constrained ISL topologies.

    \item We develop an efficient first-order projected gradient ascent method %
    that approximately solves the first-stage optimization problem and scales to %
    large constellations. %

\end{enumerate}

Sec.~\ref{sec:related_work} reviews related work, Secs.~\ref{sec:system_model} and~\ref{sec:problem_formulation} present the system model and problem formulation, Sec.~\ref{sec:framework} develops the spectral framework, Sec.~\ref{sec:baseline} describes the local-search baseline, Sec.~\ref{sec:simulation} presents simulation results, and Sec.~\ref{sec:conclusion} concludes the paper.

\section{Related Work}
\label{sec:related_work}

Prior work on ISL topology design has explored time-division methods for periodically activating feasible links for long-term connectivity~\cite{chu2018time}, link scheduling for navigation constellations~\cite{sun2022inter}, and multi-objective assignment algorithms that account for payload and visibility constraints. Large-scale ISL scheduling has also been modeled as a discrete network multi-commodity flow problem, with data-driven search heuristics proposed to minimize transmission delay~\cite{liu2022data}. While these approaches emphasize temporal reachability, they do not explicitly optimize worst-case latency or network diameter.

Multicast and broadcast strategies in LEO constellations have also been explored, including core-based multicast trees and disruption-tolerant mechanisms for multilayered satellite networks~\cite{cheng2007core, zheng2010routing}. Surveys of routing in LEO networks highlight trade-offs between centralized and distributed protocols and the need to adapt to predictable topological changes~\cite{xiaogang2016survey}. Centralized, software-defined networking approaches precompute routing tables for anticipated topologies, reducing on-board computation and routing convergence time~\cite{corici2024sdn}. While improving resilience and group communication, these methods often rely on idealized routing assumptions and ignore geometric constraints inherent to LEO dynamics.

Stochastic geometry has been used to model spatial connectivity and evaluate average-case latency~\cite{wang2022stochastic}, and congestion-aware routing protocols improve throughput under dynamic traffic~\cite{dai2021distributed}. Analytic models estimate ISL hop counts between ground users~\cite{chen2021analysis}, and algorithms have been proposed to compute shortest-distance paths directly~\cite{chen2024shortest}. In~\cite{mollakhani2026inter}, we %
proposed a local-search heuristic to refine inter-plane link assignments in Starlink-inspired constellations, but the method is prone to local optima as constellation size grows. %

In spectral graph theory,
the algebraic connectivity of the Laplacian captures network expansion %
and is related to diameter through spectral bounds~\cite{chung1997spectral, chung1989diameters}, and maximizing this spectral gap via convex optimization can eliminate bottlenecks~\cite{ghosh2006growing}. Continuous relaxations of discrete link selection problems, supported by semidefinite programming (SDP) and convex analysis~\cite{fan1949theorem, 
lewis1996convex, 
vandenberghe1996semidefinite 
}, 
enable tractable optimization of low-diameter networks. %
Leveraging these %
techniques,
the proposed framework addresses diameter minimization via %
a spectral surrogate,
which scheduling and routing approaches~\cite{chu2018time,sun2022inter,liu2022data,corici2024sdn} do not target; optimizes over arbitrary feasible-link sets rather than fixed constructions such as +Grid~\cite{bhattacherjee2019network} or common-offset designs~\cite{rao2025minimum}; and refines the local-search solution of~\cite{mollakhani2026inter} through a global spectral relaxation.

\section{System Model}
\label{sec:system_model}

We model the LEO network with an epoch-based framework: the set of active ISLs is fixed within each epoch and updated only at epoch boundaries, so routing and link-assignment decisions are deterministic within an epoch.

\subsection{Constellation Geometry}

We consider a Walker--Delta constellation with $N_p$ orbital planes and $N_s$ satellites per plane, for a total of $N = N_p \times N_s$ satellites. %
All satellites orbit at altitude %
$h$ %
with %
inclination %
$\theta$, %
giving orbital radius %
$r = R_E + h$, where $R_E = 6{,}371$~km is Earth's mean radius.
The planes are evenly spaced in longitude, and satellites within each plane are uniformly distributed. %
To model deviations from the nominal geometry due to orbit-insertion errors and gravitational perturbations~\cite{hu2025station}, each plane $i$ receives a phase offset $\phi_i$ drawn uniformly from $[0, \phi_{\max}]$.

\subsection{Link Feasibility Constraints}

Let $\bm{r}_u(t) \in \mathbb{R}^3$ denote the position of satellite $u$ at time $t$ in an Earth-centered Cartesian frame.
The maximum %
ISL distance, %
$d_{\max}$, is set %
by physical-layer factors %
including laser power, beam divergence, and receiver sensitivity. A candidate link between %
satellites $u$ and $v$ is %
feasible at time $t$ if %
\begin{align}
    \| \bm{r}_u(t) - \bm{r}_v(t) \| \leq d_{\max}, \label{eq:distance_constraint}
\end{align}
where $d_{\max}< 2\sqrt{2hR_E+h^2}$ %
ensures line of sight.

\subsection{Snapshot
vs.\ Viability-Constrained Models}

Let $V$ denote the set of all $N$ satellites, and let $\binom{V}{2}=\{\{i,j\}: i,j\in V,\ i\neq j\}$ denote the set of all unordered pairs of satellites. %
Let $\Esnapshot(t)$ denote the set of pairs satisfying~\eqref{eq:distance_constraint} at time $t$. We consider two %
link feasibility models:

\begin{definition}[Snapshot
Model]
\label{def:snapshot}
A link %
 is %
feasible if it satisfies%
~\eqref{eq:distance_constraint} at a single instant in time $t_0$ (an instantaneous epoch).
The feasible set is $\Esnapshot = \Esnapshot(t_0)$.
\end{definition}

\begin{definition}[Viability-Constrained Model]
\label{def:viability}
A link %
is feasible %
if it satisfies%
~\eqref{eq:distance_constraint} %
through an orbital period $T$ %
(a full-orbit epoch). 
The %
feasible set is $\Eviable = \bigcap_{t \in [t_0,\, t_0+T]} \Esnapshot(t)$, so $\Eviable \subseteq \Esnapshot$.
\end{definition}

This distinction matters because optical ISLs are slow to reconfigure: retargeting requires physically steering the laser terminal and reacquiring beam alignment, which takes seconds~\cite{bhattacharjee2023laser}. Topologies that remain valid longer therefore need fewer reconfigurations.

Table~\ref{tab:notation} summarizes the principal notation used throughout the paper.
\begin{table*}[!t]
\centering
\caption{Principal notation.}
\label{tab:notation}
\scriptsize
\setlength{\tabcolsep}{2.5pt}
\renewcommand{\arraystretch}{0.96}
\begin{tabular}{c p{0.35\textwidth} c p{0.35\textwidth}}
\toprule
\textbf{Symbol} & \textbf{Meaning} & \textbf{Symbol} & \textbf{Meaning} \\
\midrule
\multicolumn{4}{l}{\textbf{Constellation and feasibility}} \\
$N_p,N_s,N$ & Orbital planes, satellites per plane, and total satellites & $h,\theta$ & Altitude and inclination \\
$R_E,r$ & Earth and orbital radii ($r=R_E+h$) & $\phi_i,\phi_{\max}$ & Plane-$i$ phase offset and its maximum \\
$\bm{r}_u(t)$ & Position of satellite $u$ at time $t$ & $t_0,T$ & Snapshot time and orbital period \\
$d_{\max}$ & Maximum ISL distance & $\Esnapshot,\Eviable$ & Snapshot- and viability-feasible edge sets \\
\multicolumn{4}{l}{\textbf{Graph formulation}} \\
$V$ & Satellite set & $\Epotential$, $M$ & Candidate edges, chosen as $\Esnapshot$ or $\Eviable$, $M=|F|$ \\
$E$ & Selected edges, $E\subseteq\Epotential$ & $D$ & Per-satellite degree budget \\
$d_E(u,v)$ & Hop distance & $\operatorname{diam}(E)$ & Network diameter \\
\multicolumn{4}{l}{\textbf{Spectral method}} \\
$\bm{x}\in[0,1]^M,\bm{x}^*$ & Relaxed link strengths; $\bm{x}^*$ is returned by Stage~1 & $\bm{y}\in\{0,1\}^M$ & Binary link-selection vector \\
$E(\bm{y})$ & Topology selected by $\bm{y}$ & $\bm{L}(\bm{x}),\bm{L}(\bm{y})$ & Weighted and unweighted Laplacians \\
$\bm{b}_e$, $\bm{B}$ & Incidence vector of link $e$, incidence matrix %
& $\mathcal I(i)$ & Candidate links incident to satellite $i$ \\
$\lambda_k(\cdot)$ & $k$th-smallest eigenvalue; $\lambda_2$ is algebraic connectivity & $f(\bm{x})$ & Spectral objective; $f(\bm{x})=\lambda_2(\bm{L}(\bm{x}))$ for $\bm{x}\geq\bm{0}$ \\
$\bm{q}(\bm{v})$ & Vector with entries $q_e(\bm{v})=(\bm{b}_e^T\bm{v})^2$ & $\mathcal S$ & Unit vectors orthogonal to $\bm{1}$ \\
\bottomrule
\end{tabular}
\end{table*}

\section{Problem Formulation}
\label{sec:problem_formulation}

For the epoch under consideration, let $\Epotential\subseteq\binom{V}{2}$ denote the set of feasible undirected ISLs and $E\subseteq\Epotential$ the set of active ISLs; since $V$ is fixed, we index graph quantities by $E$ alone. Each satellite has $D$ bidirectional optical terminals for inter-satellite communication, so
\begin{align}
    \deg_E %
    (v) \leq D, \quad \forall v \in V. \label{eq:degree_bound}
\end{align}

Let $d_E(u,v)$ denote the hop distance between $u$ and $v$ in $(V,E)$, that is, the minimum number of links on a path between them, with $d_E(u,v) = +\infty$ if no such path exists.

\begin{definition}[Network Diameter]
    The network diameter of edge set $E$
    is %
    \begin{align}
    \operatorname{diam}(E) %
    = \max_{u, v \in V} d_E(u, v). \label{eq:diameter}
    \end{align}
\end{definition}

The topology design problem is
to select a %
subset of edges $E%
\subseteq\Epotential$
that minimizes the network diameter:
\begin{subequations}\label{eq:discrete_problem}
\begin{align}
    \minimize_{E\subseteq\Epotential} \quad & \operatorname{diam}(E) %
    \label{eq:obj_general} \\
    \subjectto \quad & \deg_{E}%
    (v) \leq D, \quad \forall v \in V .\label{eq:deg_general}
\end{align}
\end{subequations}

The candidate set is $\Epotential = \Esnapshot$ under the snapshot model (Definition~\ref{def:snapshot}) and $\Epotential = \Eviable$ under the viability-constrained model (Definition~\ref{def:viability}).%

Problem~\eqref{eq:discrete_problem} is NP-hard for general $\Epotential$: deciding whether some feasible $E$ has finite diameter amounts to deciding whether $(V, \Epotential)$ has a connected spanning subgraph with maximum degree at most $D$, which is NP-complete for every fixed $D \geq 2$~\cite{GareyJohnson1979}.

We address~\eqref{eq:discrete_problem} with the spectral framework of Sec.~\ref{sec:framework} and compare it against the iterative local-search baseline %
(Sec.~\ref{sec:baseline}).

\section{A Topology Design Framework}
\label{sec:framework}

Since~\eqref{eq:discrete_problem} is NP-hard in general, we approximate it in two stages. Stage~1 approximately solves a convex relaxation that maximizes the algebraic connectivity, a spectral surrogate for the diameter (Secs.~\ref{subsec:spectral_surrogate} and~\ref{subsec:gradient_opt}). Stage~2 rounds the resulting fractional solution to a degree-feasible topology by integer linear programming (ILP) (Sec.~\ref{subsec:two_stage}).

\subsection{Spectral Relaxation}
\label{subsec:spectral_surrogate}

We assume throughout that $N \geq 2$ and that the candidate graph $(V, \Epotential)$ is connected; otherwise every $E \subseteq \Epotential$ is disconnected, and~\eqref{eq:discrete_problem} has no finite-diameter solution. For a symmetric matrix $\bm{A} \in \mathbb{R}^{N \times N}$, let $\lambda_1(\bm{A}) \leq \dots \leq \lambda_N(\bm{A})$ denote its eigenvalues.

Let $M = |\Epotential|$. Fixing an ordering of the links in $\Epotential$, we write a vector $\bm{x} \in \mathbb{R}^M$ as $\bm{x} = (x_e)_{e \in \Epotential}$, so that $x_e$ is the entry associated with link $e$. For each $e = \{i,j\} \in \Epotential$, fix an arbitrary orientation $i \to j$ and define the incidence vector $\bm{b}_e = \bm{e}_i - \bm{e}_j$, where $\bm{e}_i \in \mathbb{R}^N$ is the $i$th standard basis vector. Let $\bm{B} \in \mathbb{R}^{N \times M}$ be the incidence matrix with columns $\bm{b}_e$, and let $\mathcal{I}(i) \subseteq \Epotential$ denote the set of feasible links incident to satellite~$i$. For a vector of link weights $\bm{x} \in \mathbb{R}^M$, define
\begin{align}
    \bm{L}(\bm{x}) = \sum_{e \in \Epotential} x_e \bm{b}_e \bm{b}_e^T = \bm{B} \diag(\bm{x}) \bm{B}^T.
    \label{eq:laplacian}
\end{align}
For $e = \{i,j\}$, the matrix $\bm{b}_e \bm{b}_e^T$ has entries $1$ at $(i,i)$ and $(j,j)$, $-1$ at $(i,j)$ and $(j,i)$, and zeros elsewhere. Hence the $i$th diagonal entry of $\bm{L}(\bm{x})$ is the node strength $\sum_{e \in \mathcal{I}(i)} x_e$, and for $i \neq j$, the $(i,j)$ entry is $-x_e$ if $e = \{i,j\} \in \Epotential$ and zero otherwise; for $\bm{x} \geq \bm{0}$, $\bm{L}(\bm{x})$ is thus the weighted Laplacian of $(V, \Epotential)$ with link weights $\bm{x}$. Negating $\bm{b}_e$ leaves $\bm{b}_e \bm{b}_e^T$ unchanged, so $\bm{L}(\bm{x})$ does not depend on the chosen orientations. Since $\bm{b}_e^T \bm{1} = 0$ for every $e$, $\bm{L}(\bm{x})\bm{1} = \bm{0}$ for all $\bm{x}$. For $\bm{v} \in \mathbb{R}^N$, define $\bm{q}(\bm{v}) \in \mathbb{R}^M$ by
\begin{align} \label{eq:qev}
q_e(\bm{v}) = (\bm{b}_e^T \bm{v})^2 = (v_i - v_j)^2
\end{align}
for $e = \{i,j\}$, so that
\begin{align}
    \bm{v}^T \bm{L}(\bm{x}) \bm{v} = \bm{q}(\bm{v})^T \bm{x}.
    \label{eq:quadratic_form}
\end{align}

A topology $E \subseteq \Epotential$ corresponds to the binary selection vector $\bm{y} \in \{0,1\}^M$ with $y_e = 1$ if and only if $e \in E$; we write $E(\bm{y}) = \{e \in \Epotential : y_e = 1\}$. Then $\bm{L}(\bm{y})$ is the unweighted Laplacian of $(V, E(\bm{y}))$, and $\deg_{E(\bm{y})}(i) = \sum_{e \in \mathcal{I}(i)} y_e$. We relax $\bm{y}$ to a vector of link strengths $\bm{x} \in [0,1]^M$ and replace the diameter with the algebraic connectivity $\lambda_2(\bm{L}(\bm{x}))$:
\begin{subequations}
\label{eq:spectral_problem}
\begin{align}
    \maximize_{\bm{x} \in [0,1]^M} \quad & \lambda_2(\bm{L}(\bm{x})) \\
    \subjectto \quad & \sum_{e \in \mathcal{I}(i)} x_e \leq D, \quad \forall i \in V.
    \label{eq:continuous_budget}
\end{align}
\end{subequations}

The algebraic connectivity is a surrogate for the diameter, not an equivalent. %
If $(V, E(\bm{y}))$ is connected with maximum degree $\Delta$, then~\cite{alon1985lambda}
\begin{align}
    \operatorname{diam}(E(\bm{y})) \leq 2\left\lceil \sqrt{\frac{2\Delta}{\lambda_2(\bm{L}(\bm{y}))}}\,\log_2 N \right\rceil.
    \label{eq:diameter_spectral_bound}
\end{align}
The right-hand side is nondecreasing in $\Delta$, so~\eqref{eq:diameter_spectral_bound} also holds with $\Delta$ replaced by $D$ whenever $E(\bm{y})$ satisfies~\eqref{eq:degree_bound}. For fixed $N$ and $D$, a larger $\lambda_2$ therefore gives a smaller upper bound on the diameter; it does not guarantee a diameter-optimal topology. Whereas~\eqref{eq:diameter_spectral_bound} involves the algebraic connectivity $\lambda_2(\bm{L}(\bm{y}))$ of a binary topology,~\eqref{eq:spectral_problem} maximizes that of a fractional $\bm{x}$. %
Because every degree-feasible $\bm{y}$ is feasible for~\eqref{eq:spectral_problem}, the optimal value of~\eqref{eq:spectral_problem} upper-bounds $\lambda_2(\bm{L}(\bm{y}))$; Proposition~\ref{prop:spectral-rounding} (Sec.~\ref{subsec:two_stage}) gives a lower bound in terms of $\lambda_2(\bm{L}(\bm{x}))$ and the link weight discarded in rounding.

Let $\mathcal{S} = \{\bm{v} \in \mathbb{R}^N : \|\bm{v}\| = 1,\ \bm{v}^T \bm{1} = 0\}$, and define, for all $\bm{x} \in \mathbb{R}^M$,
\begin{align}
    f(\bm{x}) = \min_{\bm{v} \in \mathcal{S}} \bm{q}(\bm{v})^T \bm{x}.
    \label{eq:f}
\end{align}

\begin{proposition}[Convexity]
\label{prop:convexity}
The function $f$ is concave on $\mathbb{R}^M$, and %
for every $\bm{x} \geq \bm{0}$,
\begin{align}
\label{eq:f=lambda2}
f(\bm{x}) = \lambda_2(\bm{L}(\bm{x})) .%
\end{align}
Consequently,~\eqref{eq:spectral_problem} maximizes a concave function over a convex polytope and is a convex optimization problem.
\end{proposition}

\begin{proof}
As a pointwise minimum of linear functions of $\bm{x}$, $f$ is concave~\cite{boyd2004convex}. For $\bm{x} \geq \bm{0}$,~\eqref{eq:quadratic_form} gives $\bm{v}^T \bm{L}(\bm{x}) \bm{v} \geq 0$ for all $\bm{v}$, so $\bm{L}(\bm{x}) \succeq 0$; together with $\bm{L}(\bm{x})\bm{1} = \bm{0}$, this gives $\lambda_1(\bm{L}(\bm{x})) = 0$ with eigenvector $\bm{1}$. The Courant--Fischer theorem~\cite{hornjohnson2012matrix} then gives $\lambda_2(\bm{L}(\bm{x})) = \min_{\bm{v} \in \mathcal{S}} \bm{v}^T \bm{L}(\bm{x}) \bm{v}$, which equals $f(\bm{x})$ by~\eqref{eq:quadratic_form}. The feasible set of~\eqref{eq:spectral_problem} is the intersection of finitely many closed half-spaces and is bounded, so it is a convex polytope.
\end{proof}

Problem~\eqref{eq:spectral_problem} can also be written as an SDP
by maximizing $\tau$ subject to $\lambda_2(\bm{L}(\bm{x})) \geq \tau$ and expressing this constraint as a linear matrix inequality in $(\bm{x}, \tau)$.
Write any $\bm{v} \in \mathbb{R}^N$ as $\bm{v} = a\bm{1} + \bm{w}$ with $\bm{w} \perp \bm{1}$. Then $\bm{v}^T \bm{L}(\bm{x}) \bm{v} = \bm{w}^T \bm{L}(\bm{x}) \bm{w}$, because $\bm{L}(\bm{x})\bm{1} = \bm{0}$, and $\bm{v}^T (\bm{I} - \tfrac{1}{N}\bm{1}\bm{1}^T) \bm{v} = \|\bm{w}\|^2$. Hence $\bm{L}(\bm{x}) \succeq \tau (\bm{I} - \tfrac{1}{N}\bm{1}\bm{1}^T)$ holds if and only if $\bm{w}^T \bm{L}(\bm{x}) \bm{w} \geq \tau \|\bm{w}\|^2$ for all $\bm{w} \perp \bm{1}$, that is, if and only if $f(\bm{x}) \geq \tau$, where $f(\bm{x}) = \lambda_2(\bm{L}(\bm{x}))$ for feasible $\bm{x}$ by Proposition~\ref{prop:convexity}. Problem~\eqref{eq:spectral_problem} is therefore equivalent to the SDP~\cite{boyd2004convex}
\begin{subequations}
\label{eq:sdp}
\begin{align}
    \maximize_{\bm{x} \in [0,1]^M}
    \quad & \tau \\
    \subjectto \quad & \bm{L}(\bm{x}) \succeq \tau \bigl(\bm{I} - \tfrac{1}{N}\bm{1}\bm{1}^T\bigr), \\
    & \sum_{e \in \mathcal{I}(i)} x_e \leq D, \quad \forall i \in V.
\end{align}
\end{subequations}

\subsection{First-Order Solver}
\label{subsec:gradient_opt}

Although~\eqref{eq:spectral_problem} has only $M$ variables, interior-point methods, which solve the SDP~\eqref{eq:sdp} to high accuracy, form and factor dense matrices of sizes $N \times N$ and $(M+1) \times (M+1)$ at every iteration~\cite{vandenberghe1996semidefinite}, which becomes prohibitive in time and memory for large constellations. We therefore use projected gradient ascent (PGA)~\cite{bertsekas1999nonlinear}, which works directly with $\bm{x}$ and the sparse matrix $\bm{L}(\bm{x})$. In our PGA, the gradient is replaced by an $\varepsilon$-supergradient, the box constraints are enforced by projection %
and the degree constraints by a quadratic penalty.

\subsubsection{Penalty Formulation}
Let $\bm{s}_i \in \{0,1\}^M$ denote the $i$th column of $|\bm{B}^T|$, the entrywise absolute value of $\bm{B}^T$, so that $\bm{s}_i^T \bm{x} = \sum_{e \in \mathcal{I}(i)} x_e$ is the strength of node~$i$. To penalize violations of the degree constraints~\eqref{eq:continuous_budget}, %
define the penalty
\begin{align}
    \Phi(\bm{x}) = \sum_{i \in V} \bigl[(\bm{s}_i^T \bm{x} - D)^+\bigr]^2,
    \label{eq:penalty}
\end{align}
where $(\cdot)^+ = \max(0, \cdot)$. As a sum of squared ramps of affine functions, $\Phi$ is convex and continuously differentiable, with
\begin{align}
    \nabla \Phi(\bm{x}) = 2 \sum_{i \in V} (\bm{s}_i^T \bm{x} - D)^+\, \bm{s}_i,
    \label{eq:penalty_grad}
\end{align}
i.e., $[\nabla \Phi(\bm{x})]_e = 2(\bm{s}_i^T \bm{x} - D)^+ + 2(\bm{s}_j^T \bm{x} - D)^+$ for $e = \{i,j\}$. For a penalty parameter $\rho > 0$, the penalized problem is
\begin{align}
    \maximize_{\bm{x} \in [0,1]^M} \quad \mathcal{J}(\bm{x}; \rho) = f(\bm{x}) - \rho\, \Phi(\bm{x}),
    \label{eq:penalized_obj}
\end{align}
whose objective is concave, since $f$ is concave by Proposition~\ref{prop:convexity} and $\Phi$ is convex. Since $\mathcal{J}(\cdot\,; \rho)$ is continuous and $[0,1]^M$ is compact, maximizers of~\eqref{eq:penalized_obj} exist. For any sequence $\rho_k \to \infty$, every limit point of a corresponding sequence of maximizers solves~\eqref{eq:spectral_problem}~\cite{bertsekas1999nonlinear, nocedal1999numerical}.

\subsubsection{Supergradients}
Since $f$ is concave on $\mathbb{R}^M$ (Proposition~\ref{prop:convexity}), we use supergradients, the concave counterpart of subgradients, in place of gradients where $f$ is not differentiable. For $\bm{x} \in \mathbb{R}^M$ and $\varepsilon \geq 0$, a vector $\bm{d} \in \mathbb{R}^M$ is an \emph{$\varepsilon$-supergradient} of $f$ at $\bm{x}$ if $f(\bm{z}) \leq f(\bm{x}) + \bm{d}^T(\bm{z} - \bm{x}) + \varepsilon$ for all $\bm{z} \in \mathbb{R}^M$~\cite{rockafellar1970convex}; a $0$-supergradient is called a \emph{supergradient}. For $\bm{x} \geq \bm{0}$, write $\lambda_k = \lambda_k(\bm{L}(\bm{x}))$, and let $\bm{v}^{(1)}, \dots, \bm{v}^{(N)}$ be orthonormal eigenvectors of $\bm{L}(\bm{x})$ with $\bm{L}(\bm{x})\bm{v}^{(k)} = \lambda_k \bm{v}^{(k)}$ and $\bm{v}^{(1)} = N^{-1/2}\bm{1}$; such a basis exists because $\bm{1}$ is an eigenvector for $\lambda_1 = 0$. %

\begin{lemma}[Supergradients of the spectral objective]
\label{lm:supergradients}
Let $f$ be given by~\eqref{eq:f}, and let $\bm{x} \geq \bm{0}$.
\begin{enumerate}
    \item If $\lambda_2$ is a simple (nonrepeated) eigenvalue of $\bm{L}(\bm{x})$, then $f$ is differentiable at $\bm{x}$ and $\nabla f(\bm{x}) = \bm{q}(\bm{v}^{(2)})$.
    \item Let $\tilde{\bm{v}}^{(2)}, \dots, \tilde{\bm{v}}^{(\ell)}$, with $\ell \geq 2$, be unit vectors orthogonal to $\bm{1}$, and let $\tilde{\lambda}_k = (\tilde{\bm{v}}^{(k)})^T \bm{L}(\bm{x})\, \tilde{\bm{v}}^{(k)}$ be their Rayleigh quotients, which satisfy $\tilde{\lambda}_k \geq \lambda_2$ by the Courant--Fischer theorem. For $\varepsilon > 0$, let
    \begin{align}
        \mathcal{K} \subseteq \{k \in \{2, \dots, \ell\} : \tilde{\lambda}_k < \tilde{\lambda}_2 + \varepsilon\}
        \label{eq:k-cluster}
    \end{align}
    be nonempty. Then
    \begin{align}
        \bm{d}(\bm{x}) = \frac{1}{|\mathcal{K}|} \sum_{k \in \mathcal{K}} \bm{q}(\tilde{\bm{v}}^{(k)})
        \label{eq:spectral_grad}
    \end{align}
    is an $\tilde{\varepsilon}$-supergradient of $f$ at $\bm{x}$, where
    \begin{align}
        0 \leq \tilde{\varepsilon} = \frac{1}{|\mathcal{K}|} \sum_{k \in \mathcal{K}} \tilde{\lambda}_k - \lambda_2 < \varepsilon + (\tilde{\lambda}_2 - \lambda_2).
        \label{eq:eps_tilde}
    \end{align}
    If $\tilde{\bm{v}}^{(k)} = \bm{v}^{(k)}$ for all $k$, then $\tilde{\varepsilon} < \varepsilon$.
\end{enumerate}
\end{lemma}

The proof of Lemma~\ref{lm:supergradients} is given in Appendix~\ref{app:supergradients}.

Highly symmetric constellations such as Walker--Delta can have repeated or nearly repeated eigenvalues. At a repeated $\lambda_2$, the eigenvector $\bm{v}^{(2)}$ is not unique, and $f$ can be nondifferentiable~\cite{lewis1996convex}. Near a repeated $\lambda_2$, $\bm{v}^{(2)}$, and hence $\bm{q}(\bm{v}^{(2)})$, can change abruptly with $\bm{x}$, because eigenvector sensitivity grows as the eigenvalue gap shrinks~\cite{stewart1990matrix}. We therefore ascend along the averaged direction~\eqref{eq:spectral_grad} rather than taking $\mathcal{K} = \{2\}$; treating a cluster of nearly coincident eigenvalues jointly is common in eigenvalue optimization~\cite{overton1992sum}. Because $\Phi$ is convex and differentiable, $\Phi(\bm{z}) \geq \Phi(\bm{x}) + \nabla\Phi(\bm{x})^T(\bm{z} - \bm{x})$ for all $\bm{z}$. Combined with Lemma~\ref{lm:supergradients}(ii), this shows that $\bm{d}(\bm{x}) - \rho\,\nabla\Phi(\bm{x})$ is an $\tilde{\varepsilon}$-supergradient of $\mathcal{J}(\cdot\,; \rho)$ at $\bm{x}$.

\subsubsection{Momentum and Annealing}
Each iteration requires the eigenpairs indexed by $\mathcal{K}$. At iteration $t$, let $(\tilde{\lambda}_k, \tilde{\bm{v}}^{(k)})$ denote the approximation of the $k$th eigenpair of $\bm{L}(\bm{x}_t)$ computed by an iterative eigensolver. The ascent direction is
\begin{align}
    \bm{g}_t = \bm{d}(\bm{x}_t) - \rho_t \nabla \Phi(\bm{x}_t),
\end{align}
where, for these eigenpairs, $\mathcal{K}$ satisfies~\eqref{eq:k-cluster} and $\bm{d}(\bm{x}_t)$ is given by~\eqref{eq:spectral_grad}.
We add Polyak's heavy-ball momentum~\cite{polyak1964some}, increase the penalty parameter, and decrease the step size over $T_{\max}$ iterations according to
\begin{align}
    \rho_t &= \rho_{\min} + (\rho_{\max} - \rho_{\min}) 
    \,t^2\, T_{\max}^{-2}, 
    \label{eq:rho_schedule} \\
    \eta_t &=
    (1 + \alpha t)^{-1} \eta_0. 
    \label{eq:eta_schedule}
\end{align}
Let $\mu \in [0,1)$ be the momentum coefficient, and let $\Pi_{[0,1]}$ denote elementwise clipping to $[0,1]$, which is the Euclidean projection onto $[0,1]^M$. Starting from $\bm{x}_0 \in [0,1]^M$ and $\bm{m}_0 = \bm{0}$, the update is
\begin{align}
    \bm{m}_{t+1} &= \mu\, \bm{m}_t + \eta_t \bm{g}_t, \label{eq:update_rule_m} \\
    \bm{x}_{t+1} &= \Pi_{[0,1]}\bigl(\bm{x}_t + \bm{m}_{t+1}\bigr). \label{eq:update_rule_x}
\end{align}
Because only the $\ell \ll N$ smallest eigenpairs are computed, $\mathcal{K}$ may miss eigenvalues of $\bm{L}(\bm{x}_t)$ within $\varepsilon$ of $\tilde{\lambda}_2$; Lemma~\ref{lm:supergradients}(ii) does not require $\mathcal{K}$ to contain them, and it applies to the computed eigenvectors provided that each $\tilde{\lambda}_k$ is the Rayleigh quotient of $\tilde{\bm{v}}^{(k)}$.
The direction~\eqref{eq:spectral_grad} is then an $\tilde{\varepsilon}$-supergradient of $f$ at $\bm{x}_t$ with $\tilde{\varepsilon} < \varepsilon + (\tilde{\lambda}_2 - \lambda_2)$, where $\tilde{\lambda}_2 - \lambda_2 \geq 0$ is the error in the computed $\lambda_2$.

If the eigensolver fails to converge, or if it returns $\tilde{\lambda}_2 < \delta$ for a small threshold $\delta > 0$, which indicates that the support of $\bm{x}_t$ is numerically disconnected, the update is skipped; instead, a small random nonnegative perturbation is added to $\bm{x}_t$, followed by clipping to $[0,1]$. Every iterate therefore lies in $[0,1]^M$.

Without momentum and for fixed $\rho$ and $\varepsilon$, this scheme is the projected $\varepsilon$-supergradient method, whose best iterate converges to within $\varepsilon$ of the maximum of~\eqref{eq:penalized_obj} under diminishing step sizes such as~\eqref{eq:eta_schedule}~\cite{bertsekas1999nonlinear}. %
Momentum and an increasing penalty parameter are standard techniques for accelerating convergence and enforcing feasibility, respectively, but this guarantee does not cover them.
The PGA therefore returns the iterate $\bm{x}^* \in [0,1]^M$ with the largest $\lambda_2$ among those that approximately satisfy~\eqref{eq:continuous_budget}, scaled down where necessary so that it satisfies~\eqref{eq:continuous_budget}. %
With no convergence guarantee and finite $T_{\max}$ and $\rho_{\max}$, $\bm{x}^*$ is in general only an approximate solution of~\eqref{eq:spectral_problem}.
Since~\eqref{eq:spectral_problem} is itself only a surrogate for~\eqref{eq:discrete_problem} (Sec.~\ref{subsec:spectral_surrogate}), we seek an $\bm{x}^*$ with large $\lambda_2$ rather than an exact maximizer.

\subsection{Discrete Rounding}
\label{subsec:two_stage}

Stage~2 converts $\bm{x}^*$ into a binary selection vector $\bm{y}$ using a bound on the algebraic connectivity lost in rounding.

\begin{proposition}[Spectral loss under rounding]
\label{prop:spectral-rounding}
For any $\bm{x} \in [0,1]^M$ and $\bm{y} \in \{0,1\}^M$,
\begin{align}
    \lambda_2(\bm{L}(\bm{y})) \geq \lambda_2(\bm{L}(\bm{x})) - 2 \sum_{e \in \Epotential} (1 - y_e)\, x_e.
    \label{eq:rounding_spectral_loss}
\end{align}
\end{proposition}

\begin{proof}
Let
$    \bm{R}(\bm{x}, \bm{y}) = \sum_{e \in \Epotential} (1 - y_e)\, x_e\, \bm{b}_e \bm{b}_e^T $
be the Laplacian contribution of the link weight that $\bm{y}$ discards. Since $0 \leq x_e \leq 1$ and $y_e \in \{0,1\}$, both $\bm{R}(\bm{x},\bm{y})$ and
\begin{align}
    \bm{L}(\bm{y}) - \bigl(\bm{L}(\bm{x}) - \bm{R}(\bm{x},\bm{y})\bigr) = \sum_{e \in \Epotential} y_e (1 - x_e)\, \bm{b}_e \bm{b}_e^T
    \label{eq:rounding_difference}
\end{align}
are nonnegative combinations of the matrices $\bm{b}_e \bm{b}_e^T \succeq 0$ and hence positive semidefinite. Thus $\bm{L}(\bm{y}) \succeq \bm{L}(\bm{x}) - \bm{R}(\bm{x},\bm{y})$, and we have
\begin{align}
    \lambda_2(\bm{L}(\bm{y}))
    &\geq \lambda_2\bigl(\bm{L}(\bm{x}) - \bm{R}(\bm{x},\bm{y})\bigr) \label{eq:weyl_monotonicity} \\
    &\geq \lambda_2(\bm{L}(\bm{x})) - \lambda_N\bigl(\bm{R}(\bm{x},\bm{y})\bigr) \label{eq:weyl_additive} \\
    &\geq \lambda_2(\bm{L}(\bm{x})) - \operatorname{tr}\bigl(\bm{R}(\bm{x},\bm{y})\bigr) \label{eq:trace_bound}
\end{align}
where~\eqref{eq:weyl_monotonicity} follows from Weyl's monotonicity theorem,~\eqref{eq:weyl_additive} from Weyl's additive eigenvalue bound~\cite{hornjohnson2012matrix},
and~\eqref{eq:trace_bound} follows from the nonnegativity of the eigenvalues of $\bm{R}(\bm{x},\bm{y}) \succeq 0$ and the trace--eigenvalue identity~\cite{hornjohnson2012matrix}. Moreover,
    $\operatorname{tr}\bigl(\bm{R}(\bm{x},\bm{y})\bigr) 
     = \sum_{e \in \Epotential} (1 - y_e)\, x_e \operatorname{tr}(\bm{b}_e \bm{b}_e^T) %
    = 2\sum_{e \in \Epotential} (1 - y_e)\, x_e$,
where $\operatorname{tr}(\bm{b}_e \bm{b}_e^T) = \bm{b}_e^T\bm{b}_e = 2$ because $\bm{b}_e = \bm{e}_i - \bm{e}_j$. Combining this and~\eqref{eq:trace_bound} %
proves~\eqref{eq:rounding_spectral_loss}.
\end{proof}

Since Proposition~\ref{prop:spectral-rounding} holds for every $\bm{x} \in [0,1]^M$, Stage~2 does not require $\bm{x}^*$ to solve~\eqref{eq:spectral_problem}. With $\bm{x} = \bm{x}^*$ fixed, the right-hand side of~\eqref{eq:rounding_spectral_loss} depends on $\bm{y}$ only through the discarded weight $\sum_{e} (1 - y_e)\, x_e^* = \sum_{e} x_e^* - \sum_{e} x_e^*\, y_e$. Stage~2 therefore selects the degree-feasible rounding that minimizes the discarded weight:
\begin{subequations}
\label{eq:ilp}
\begin{align}
    \maximize_{\bm{y} \in \{0,1\}^M} \quad & \sum_{e \in \Epotential} x_e^*\, y_e \\
    \subjectto \quad & \sum_{e \in \mathcal{I}(i)} y_e \leq D, \quad \forall i \in V.
\end{align}
\end{subequations}

We use Proposition~\ref{prop:spectral-rounding} as a design rationale for the objective of~\eqref{eq:ilp}, not as a performance guarantee. The ILP does not enforce connectivity; the connectivity and diameter of the resulting topology are evaluated numerically in Sec.~\ref{sec:simulation}.

\subsection{Computational Complexity}
\label{subsec:complexity}

At each PGA iteration, the $\ell$ smallest eigenpairs of the sparse Laplacian $\bm{L}(\bm{x}_t)$ are computed approximately by an iterative eigensolver that accesses $\bm{L}(\bm{x}_t)$ only through matrix--vector products and stores a fixed number of vectors in $\mathbb{R}^N$. %
Each eigensolver step consists of one product with $\bm{L}(\bm{x}_t)$, which has at most $N + 2M$ nonzero entries by~\eqref{eq:laplacian}, plus $O(N)$ work on the stored vectors. Forming $\bm{d}(\bm{x}_t)$ and $\nabla\Phi(\bm{x}_t)$ requires a few products with $\bm{B}^T$ and with $|\bm{B}^T|$ and its transpose. Since $M \geq N - 1$ for a connected $(V, \Epotential)$, each of these operations costs $O(M)$, and the memory per iteration is $O(M)$, rather than the $O(N^2)$ required by dense $N \times N$ matrices. Hence one PGA iteration costs $O(n_{\mathrm{mv}} M)$ time and Stage~1 costs $O(T_{\max}\, n_{\mathrm{mv}} M)$, where $n_{\mathrm{mv}}$ bounds the number of matrix--vector products per eigensolver call and depends on the solver tolerance and the eigenvalue gaps. Stage~2 is a maximum-weight simple $b$-matching problem with $b_i = D$ for all $i$, which is solvable in polynomial time~\cite{schrijver2003combinatorial}.

\subsection{Specialization: Fixed Intra-Plane Links}
\label{subsec:fixed_intra}

In some constellations, each satellite carries dedicated terminals for its two neighbors in the same orbital plane, so these links are always active and not subject to optimization, as in Starlink~\cite{fcc21, wang1993structural}. %
Let $E_{\text{intra}}$ denote the set of links between adjacent satellites in the same plane, and let $E_{\text{inter}} \subseteq \Epotential$ denote the set of feasible links between satellites in different planes. We assume that $E_{\text{intra}} \subseteq \Epotential$, that each plane contains at least three satellites, so that every satellite is incident to exactly two links in $E_{\text{intra}}$, that $D \geq 2$, and that $(V, E_{\text{intra}} \cup E_{\text{inter}})$ is connected. The results of this section then apply with $\Epotential$ replaced by $E_{\text{intra}} \cup E_{\text{inter}}$ and with the affine constraints $x_e = 1$ for all $e \in E_{\text{intra}}$
added to~\eqref{eq:spectral_problem}, which therefore remains convex. The intra-plane links then contribute a constant term to $\bm{L}(\bm{x})$, the budget of each satellite for inter-plane links becomes $D - 2$, and Stage~1 optimizes only the $|E_{\text{inter}}|$ inter-plane strengths. Stage~2 solves~\eqref{eq:ilp} with $y_e = 1$ for all $e \in E_{\text{intra}}$. %

\section{%
An Iterative Local-Search Algorithm}
\label{sec:baseline}

We now describe a variant of the local-search heuristic of~\cite{mollakhani2026inter}, which serves as our baseline. It uses the same candidate set $\Epotential$ and degree budget $D$ as the spectral framework. Unlike in~\cite{mollakhani2026inter}, intra-plane links are not kept permanently active but are treated like any other candidate link.

\begin{algorithm}[t]
\caption{Iterative Local-Search Heuristic} %
\label{alg:isl_optimization}
\begin{algorithmic}[1]
\State \textbf{Input:} Satellite set $V$, feasible link set $\Epotential$, degree budget $D$, %
iterations $n$, repair interval $K$, perturbation size $m$
\State \textbf{Output:} Optimized topology $E^*$
\vspace{0.3em}
\State $E^* \gets$ initial topology according to Section~\ref{subsec:initial_construction}
\For{$i = 1$ to $n$}
    \State $\tilde{E} \gets E^*$
    \If{$i \bmod K = 0$} \Comment{Repair Phase}
        \State $U \gets \{v \in V : \deg_{\tilde{E}}(v) < D\}$
        \For{each satellite $u \in U$}
            \State Add feasible links to valid candidates
            \State Enforce degree constraints at both endpoints
        \EndFor
    \Else \Comment{Random Replacement Phase}
        
        \For{each $u$ in a random subset of $m$ satellites}
            \State Replace one random link $\{u,v\}$ (if any) by
            \State feasible link $\{u,w\}$ %
        \EndFor
        
    \EndIf
    \If{$\tilde{E}$ satisfies (a), (b), or (c) of Sec.~\ref{subsubsec:acceptance}}
        \State $E^* \gets \tilde{E}$
    \EndIf
\EndFor
\State \Return $E^*$
\end{algorithmic}
\end{algorithm}

\subsection{Initial Topology Construction}
\label{subsec:initial_construction}

Starting from no links, the algorithm greedily builds a degree-balanced initial topology $E_0 \subseteq \Epotential$, repeating the following steps over multiple passes to use as much of the degree budget as possible:
\begin{enumerate}
    \item The satellites with degree less than $D$ are sorted by residual capacity, so that satellites with fewer links are served first.

    \item For each such satellite $u$ in turn, the candidates are the satellites $v$ with $\{u,v\} \in \Epotential \setminus E_0$ and $\deg_{E_0}(v) < D$.

    \item To favor long links, which reduce the diameter, the candidates are sorted by distance and split into a far half and a near half, each of which is shuffled. The algorithm then attempts to link $u$ to the far candidates, in their shuffled order, and then to the near candidates, until $u$ reaches degree $D$ or no candidates remain.

\end{enumerate}

If $E_0$ is disconnected, we attempt to bridge its components through feasible link replacements, subject to the degree bound $D$. In all our experiments, $E_0$ was connected.

\subsection{Iterative Topology Optimization}
\label{subsec:iterative_optimization}

Algorithm~\ref{alg:isl_optimization} summarizes the procedure. After initialization, it runs $n$ iterations that alternate between two phases.

\subsubsection{Repair phase (every \texorpdfstring{$K$}{K} iterations)}
Every $K$th iteration, a repair phase addresses connectivity gaps: for each deficient satellite, that is, one with degree less than $D$, the algorithm attempts to add links to valid candidates with spare capacity, shuffling the candidate list for fairness.
A repaired topology replaces the best-found topology $E^*$ only if it yields a better network under the acceptance criteria of Sec.~\ref{subsubsec:acceptance}.

\subsubsection{Random replacement phase (default iterations)}
In all other iterations, each of $m$ randomly chosen satellites drops one of its links (if any), chosen at random, and tries to link instead to a different, randomly chosen valid candidate.
Because up to $m$ links change at once, each iteration explores a larger neighborhood than a single link replacement.

\subsubsection{Evaluation and acceptance}
\label{subsubsec:acceptance}
After each iteration, the modified topology $\tilde{E}$ is evaluated and compared against the best-found topology $E^*$ using an ordered multi-objective criterion.
Let $\bar{H}(E) = \frac{1}{N}\sum_{u \in V} \max_{v \in V} d_E(u,v)$ denote the average maximum hop count (mean eccentricity), and let $\sigma(E) = |E \cap \Eviable| / |E|$ denote the fraction of active links that remain feasible over the orbital period.
The new topology $\tilde{E}$ replaces $E^*$ if:
\begin{enumerate}[label=(\alph*)]
    \item $\operatorname{diam}(\tilde{E}) < \operatorname{diam}(E^*)$; or
    \item $\operatorname{diam}(\tilde{E}) = \operatorname{diam}(E^*)$ and $\bar{H}(\tilde{E}) < \bar{H}(E^*)$; or
    \item $\operatorname{diam}(\tilde{E}) = \operatorname{diam}(E^*)$, $\bar{H}(\tilde{E}) = \bar{H}(E^*)$, and $\sigma(\tilde{E}) > \sigma(E^*)$.
\end{enumerate}
Under $\Epotential = \Eviable$, every topology has $\sigma = 1$, so criterion~(c) applies only in the snapshot variant.

Algorithm~\ref{alg:isl_optimization} is a stochastic hill-climbing method: it accepts only moves that strictly improve the ordered criterion (a)--(c), and it can therefore stall at local optima. Unlike zero-temperature Metropolis--Hastings~\cite{metropolis1953equation,hastings1970monte}, it also rejects moves that leave the objective unchanged. Its distance from the optimum is %
measured via the lower bound of Sec.~\ref{subsec:lower_bound}.
Algorithm~\ref{alg:isl_optimization} also provides the warm start for the PGA (Sec.~\ref{subsec:setup}).

\section{Simulation and Results}
\label{sec:simulation}

We evaluate the performance of the proposed two-stage spectral framework against the local-search baseline.

\subsection{Simulation Setup}
\label{subsec:setup}

\paragraph{\texorpdfstring{Constellation parameters}{Constellation parameters}}
The parameters for %
a 500-satellite constellation and a 1,584-satellite constellation
are summarized in Table~\ref{tab:constellation_params}. The local-search baseline runs for $n = 300$ iterations per trial, with a repair phase every $K = 15$ iterations and a perturbation size of $m = 20$ satellites.

\begin{table}[H]
\centering
\caption{Constellation and System Parameters}
\label{tab:constellation_params}
\renewcommand{\arraystretch}{1.2}
\setlength{\tabcolsep}{5pt}
\begin{tabular}{lcc}
\toprule
\textbf{Parameter} & \textbf{Small-Scale} & \textbf{Large-Scale} \\
\midrule
Total satellites ($N$) & 500 & 1,584 \\
Orbital planes ($N_p$) & 25 & 72 \\
Satellites per plane ($N_s$) & 20 & 22 \\
Orbital altitude ($h$) & 550 km & 550 km \\
Orbital inclination ($\theta$) & $53^\circ$ & $53^\circ$ \\
Degree budget ($D$) & 4 & 4 \\
Maximum ISL distance ($d_{\max}$) & 3,500 km & 2,000--3,000 km \\
Maximum orbital phase offset %
($\phi_{\max}$) & $%
\pi/4$ & $%
\pi/4$ \\
\bottomrule
\end{tabular}
\end{table}

\paragraph{\texorpdfstring{Solver configuration}{Solver configuration}}
PGA runs for $T_{\max} = 20{,}000$ iterations. The penalty parameter~\eqref{eq:rho_schedule} increases from $\rho_{\min} = 0.001$, which tolerates constraint violations early on, to $\rho_{\max} = 30$, which enforces feasibility later; the step size~\eqref{eq:eta_schedule} uses $\eta_0 = 2$ and decay $\alpha = 0.002$, and the momentum coefficient is $\mu = 0.8$. The cluster tolerance in~\eqref{eq:k-cluster} is $\varepsilon_t = 0.05|\tilde{\lambda}_2| + 10^{-4}$. At each iteration, the ARPACK Krylov eigensolver~\cite{lehoucq1998arpack} computes the $\ell = 6$ smallest eigenpairs of $\bm{L}(\bm{x}_t)$, so $\mathcal{K} \subseteq \{2, \dots, 6\}$; the disconnection threshold is $\delta = 10^{-6}$.
Stage~2 uses HiGHS~\cite{huangfu2018highs} (version 1.12.0, default settings); every ILP solve underlying Tables~\ref{tab:snapshot_results}, \ref{tab:viability_results}, and~\ref{tab:literature_comparison} and Figs.~\ref{fig:avg_hops_noviability}--\ref{fig:worst_case_viability} terminated with status \emph{Optimal} and a relative MIP gap of at most 0.0022\% (tolerance 0.01\%).

\paragraph{\texorpdfstring{Warm-start strategy}{Warm-start strategy}}
The local-search baseline is first executed to %
obtain a feasible initial topology. This discrete solution is converted into an initial strength vector $\bm{x}_0$ for the PGA: its binary selection vector plus independent perturbations drawn uniformly from $[0, 0.05]$, clipped to $[0,1]$. All results reported for the spectral framework use this warm start.

\paragraph{\texorpdfstring{Statistical methodology}{Statistical methodology}}
All experiments are repeated over 50 %
random initializations. %
For each scenario and algorithm, we report the minimum, maximum, and mean of three %
metrics:
    {the maximum hop count (diameter)},
    {the average maximum hop count} %
    $\bar{H}$, and
    the total number of edges in the resulting subgraph.
Each topology is evaluated using all-source shortest-path analysis via breadth-first search (BFS)~\cite{kozen1992depth}. We additionally report the worst distance-weighted propagation delay and the average, over source satellites, of the maximum such delay.

\subsection{Lower Bounds on the Achievable Diameter}
\label{subsec:lower_bound}
We lower-bound the optimal diameter using the candidate graph $(V,\Epotential)$, in which all feasible links are active and the degree budget %
is ignored. %
Since removing links cannot shorten any shortest path,
$\operatorname{diam}(E)\ \ge\ \operatorname{diam}(\Epotential)$ %
for all $E\subseteq\Epotential$.
The same argument bounds the average maximum hop count and the propagation delays. Because the degree budget is ignored, the bound is generally %
unattainable.

\subsection{Small-Scale Validation (\texorpdfstring{$N = 500$}{N = 500})}
\label{subsec:small_scale}

Table~\ref{tab:snapshot_results} compares the methods in the snapshot scenario. The spectral framework attains a diameter of 9 hops in all 50 runs, versus 11 for the local-search heuristic. Since the lower bound is 7 hops, the spectral diameter is within 2 hops of the optimum, compared with 4 for the local-search heuristic. Under viability constraints (Table~\ref{tab:viability_results}), the diameters are 11 and 13 against a lower bound of 10, so the spectral solution is within one hop of the optimum, compared with three for the local-search heuristic. In both scenarios, both methods use nearly the maximum of $ND/2 = 1{,}000$ links, so the diameter gains come from link placement rather than link count.

\begin{table} %
\centering
\caption{Snapshot
Results ($N = 500$
, $d_{\max}=3{,}500$~km; hops and edges: mean [min, max] over 50 runs; delay: mean $\pm$ sample standard deviation, 50 runs; lower bound: Sec.~\ref{subsec:lower_bound}%
)}
\label{tab:snapshot_results}
\scriptsize
\setlength{\tabcolsep}{2pt}
\renewcommand{\arraystretch}{0.98}
\begin{tabular}{lccc}
\toprule
\textbf{Metric} & \textbf{Spectral Framework} & \textbf{Local-Search Baseline} & \textbf{Lower Bound} \\
\midrule
Diameter (hops) & \textbf{9.00} [9, 9] & 11.00 [11, 11] & 7 \\
Avg.\ max hops & 8.10 [8.05, 8.19] & 9.93 [9.74, 10.07] & 7.00 \\
Total edges & 999.62 [999, 1000] & 996.20 [994, 999] & -- \\
Worst delay (ms) & $94.18 \pm 0.99$ & $104.17 \pm 2.12$ & 72.18 \\
Avg.\ max delay (ms) & $86.25 \pm 0.28$ & $94.35 \pm 1.29$ & 72.04 \\
\bottomrule
\end{tabular}
\end{table}

\begin{table} %
\centering
\caption{Viability-Constrained Results ($N = 500$, $d_{\max}=3{,}500$~km; hops and edges: mean [min, max] over 50 runs; delay: mean $\pm$ sample standard deviation, 50 runs; lower bound: Sec.~\ref{subsec:lower_bound}%
)}
\label{tab:viability_results}
\scriptsize
\setlength{\tabcolsep}{2pt}
\renewcommand{\arraystretch}{0.98}
\begin{tabular}{lccc}
\toprule
\textbf{Metric} & \textbf{Spectral Framework} & \textbf{Local-Search Baseline} & \textbf{Lower Bound} \\
\midrule
Diameter (hops) & \textbf{11.00} [11, 11] & 13.00 [13, 13] & 10 \\
Avg.\ max hops & 10.64 [10.64, 10.64] & 11.89 [11.64, 12.23] & 10.00 \\
Total edges & 1000.00 [1000, 1000] & 992.76 [986, 998] & -- \\
Worst delay (ms) & $87.83 \pm 0.00$ & $101.14 \pm 2.07$ & 72.59 \\
Avg.\ max delay (ms) & $83.21 \pm 0.00$ & $91.09 \pm 0.85$ & 72.22 \\
\bottomrule
\end{tabular}
\end{table}

\subsection{Large-Scale Scalability Analysis (\texorpdfstring{$N = 1{,}584$}{N = 1,584})}
\label{subsec:large_scale}

For the 1,584-satellite constellation of Table~\ref{tab:constellation_params}, we vary $d_{\max}$ from 2,000 to 3,000~km in 100~km steps; Figs.~\ref{fig:avg_hops_noviability}--\ref{fig:worst_case_viability} show the results.
At $d_{\max}=2{,}500$~km, the proposed framework reduces the mean worst-case delay from 103.31 to 90.40~ms in the snapshot scenario and from 103.38 to 90.79~ms under viability constraints. The corresponding mean average-maximum delays decrease from 94.82 to 84.70~ms and from 93.00 to 84.19~ms, respectively. %
Across all ISL ranges, the mean spectral diameter lies 2.0--4.0 hops above the %
lower bound in the snapshot scenario and 2.0--3.0 hops above it under viability constraints, whereas the local-search heuristic lies 4.2--6.0 and 4.0--5.4 hops above it, respectively. The corresponding propagation-delay bound is 72.2--73.4~ms for every configuration, close to the half-circumference propagation time $\pi r/c\approx 72.5$~ms at the orbital radius.

\begin{figure} %
\centering
\begin{tikzpicture}
\begin{axis}[
    width=\columnwidth,
    height=6.2cm,
    xlabel={Max ISL Distance (km)},
    ylabel={Hops},
    xmin=1950, xmax=3050,
    ymin=7.5, ymax=16.8,
    xtick distance=200,
    ytick distance=1,
    grid=major,
    grid style={dotted, gray!30},
    legend style={at={(0.98,0.98)}, anchor=north east, font=\scriptsize, row sep=1pt},
    legend cell align={left}
]

\addplot[color=customCyan, line width=1.2pt, mark=*, mark size=1.5pt] coordinates {
    (2000, 14.02) (2100, 13.40) (2200, 12.99) (2300, 12.41) (2400, 11.98)
    (2500, 11.45) (2600, 11.04) (2700, 10.82) (2800, 10.30) (2900, 10.04) (3000, 9.98)
};
\addlegendentry{Spectral}

\addplot[color=customCyan, fill=customCyan, fill opacity=0.12, draw=none, forget plot]
    coordinates {
    (2000, 13.92) (2100, 13.27) (2200, 12.94) (2300, 12.27) (2400, 11.95)
    (2500, 11.30) (2600, 11.01) (2700, 10.71) (2800, 10.20) (2900, 10.01) (3000, 9.96)
    (3000, 10.01) (2900, 10.09) (2800, 10.70) (2700, 10.96) (2600, 11.10)
    (2500, 11.57) (2400, 12.01) (2300, 12.52) (2200, 13.04) (2100, 13.58) (2000, 15.02)
} \closedcycle;

\addplot[color=customOrange, line width=1.2pt, mark=square*, mark size=1.5pt] coordinates {
    (2000, 15.92) (2100, 15.34) (2200, 14.87) (2300, 14.42) (2400, 14.01)
    (2500, 13.60) (2600, 13.25) (2700, 12.87) (2800, 12.53) (2900, 12.25) (3000, 11.95)
};
\addlegendentry{Local search}

\addplot[color=customOrange, fill=customOrange, fill opacity=0.12, draw=none, forget plot]
    coordinates {
    (2000, 15.69) (2100, 15.17) (2200, 14.70) (2300, 14.16) (2400, 13.84)
    (2500, 13.40) (2600, 13.06) (2700, 12.59) (2800, 12.39) (2900, 12.06) (3000, 11.82)
    (3000, 12.06) (2900, 12.46) (2800, 12.79) (2700, 13.01) (2600, 13.45)
    (2500, 13.81) (2400, 14.20) (2300, 14.62) (2200, 15.10) (2100, 15.54) (2000, 16.16)
} \closedcycle;

\addplot[color=black, line width=1pt, dashed, mark=triangle*, mark size=1.5pt, mark options={solid}] coordinates {
    (2000, 11.00) (2100, 11.00) (2200, 11.00) (2300, 10.66) (2400, 10.00) (2500, 9.97) (2600, 9.01) (2700, 9.00) (2800, 8.83) (2900, 8.00) (3000, 8.00)
};
\addlegendentry{Lower bound}

\end{axis}
\end{tikzpicture}
\caption{Avg.\ maximum hops vs.\ ISL range (snapshot scenario, $N = 1{,}584$). %
Solid lines show the mean over 50 runs; shaded bands show the observed range across those runs. The dashed line is the lower bound of Sec.~\ref{subsec:lower_bound}.}
\label{fig:avg_hops_noviability}
\end{figure}

\begin{figure} %
\centering
\begin{tikzpicture}
\begin{axis}[
    width=\columnwidth,
    height=6.2cm,
    xlabel={Max ISL Distance (km)},
    ylabel={Hops},
    xmin=1950, xmax=3050,
    ymin=9, ymax=20,
    xtick distance=200,
    ytick distance=1,
    grid=major,
    grid style={dotted, gray!30},
    legend style={at={(0.98,0.98)}, anchor=north east, font=\scriptsize, row sep=1pt},
    legend cell align={left}
]

\addplot[color=customCyan, line width=1.2pt, mark=*, mark size=1.5pt] coordinates {
    (2000, 17.14) (2100, 16.24) (2200, 15.70) (2300, 15.04) (2400, 14.56)
    (2500, 13.98) (2600, 13.53) (2700, 13.00) (2800, 12.58) (2900, 12.02) (3000, 11.52)
};
\addlegendentry{Spectral}

\addplot[color=customCyan, fill=customCyan, fill opacity=0.12, draw=none, forget plot]
    coordinates {
    (2000, 17.01) (2100, 16.08) (2200, 15.53) (2300, 14.98) (2400, 14.42)
    (2500, 13.92) (2600, 13.39) (2700, 12.95) (2800, 12.45) (2900, 11.98) (3000, 11.38)
    (3000, 11.63) (2900, 12.28) (2800, 12.69) (2700, 13.08) (2600, 13.73)
    (2500, 14.03) (2400, 14.74) (2300, 15.14) (2200, 15.85) (2100, 16.50) (2000, 17.33)
} \closedcycle;

\addplot[color=customOrange, line width=1.2pt, mark=square*, mark size=1.5pt] coordinates {
    (2000, 19.08) (2100, 18.02) (2200, 17.16) (2300, 16.81) (2400, 16.22)
    (2500, 15.76) (2600, 15.36) (2700, 15.01) (2800, 14.75) (2900, 14.31) (3000, 13.97)
};
\addlegendentry{Local search}

\addplot[color=customOrange, fill=customOrange, fill opacity=0.12, draw=none, forget plot]
    coordinates {
    (2000, 18.61) (2100, 17.76) (2200, 16.96) (2300, 16.55) (2400, 15.99)
    (2500, 15.47) (2600, 15.15) (2700, 14.78) (2800, 14.40) (2900, 14.10) (3000, 13.77)
    (3000, 14.17) (2900, 14.53) (2800, 15.01) (2700, 15.26) (2600, 15.60)
    (2500, 16.01) (2400, 16.47) (2300, 17.03) (2200, 17.47) (2100, 18.33) (2000, 19.48)
} \closedcycle;

\addplot[color=black, line width=1pt, dashed, mark=triangle*, mark size=1.5pt, mark options={solid}] coordinates {
    (2000, 15.03) (2100, 14.15) (2200, 14.00) (2300, 13.50) (2400, 12.75) (2500, 12.11) (2600, 12.00) (2700, 11.03) (2800, 10.47) (2900, 9.90) (3000, 9.25)
};
\addlegendentry{Lower bound}

\end{axis}
\end{tikzpicture}
\caption{Avg.\ maximum hops vs.\ ISL range (viability-constrained scenario, $N = 1{,}584$). %
Plotting conventions as in Fig.~\ref{fig:avg_hops_noviability}.}
\label{fig:avg_hops_viability}
\end{figure}

\begin{figure} %
\centering
\begin{tikzpicture}
\begin{axis}[
    width=\columnwidth,
    height=6.2cm,
    xlabel={Max ISL Distance (km)},
    ylabel={Hops},
    xmin=1950, xmax=3050,
    ymin=7.5, ymax=18.5,
    xtick distance=200,
    ytick distance=1,
    grid=major,
    grid style={dotted, gray!30},
    legend style={at={(0.98,0.98)}, anchor=north east, font=\scriptsize, row sep=1pt},
    legend cell align={left}
]

\addplot[color=customCyan, line width=1.2pt, mark=*, mark size=1.5pt] coordinates {
    (2000, 15.02) (2100, 14.16) (2200, 14.00) (2300, 13.10) (2400, 12.98)
    (2500, 12.02) (2600, 12.00) (2700, 11.34) (2800, 11.00) (2900, 11.00) (3000, 10.86)
};
\addlegendentry{Spectral}

\addplot[color=customCyan, fill=customCyan, fill opacity=0.12, draw=none, forget plot]
    coordinates {
    (2000, 15.00) (2100, 14.00) (2200, 14.00) (2300, 13.00) (2400, 12.00)
    (2500, 12.00) (2600, 12.00) (2700, 11.00) (2800, 11.00) (2900, 11.00) (3000, 10.00)
    (3000, 11.00) (2900, 11.00) (2800, 11.00) (2700, 12.00) (2600, 12.00)
    (2500, 13.00) (2400, 13.00) (2300, 14.00) (2200, 14.00) (2100, 15.00) (2000, 16.00)
} \closedcycle;

\addplot[color=customOrange, line width=1.2pt, mark=square*, mark size=1.5pt] coordinates {
    (2000, 17.04) (2100, 16.66) (2200, 16.00) (2300, 15.72) (2400, 15.00)
    (2500, 15.00) (2600, 14.20) (2700, 14.00) (2800, 13.86) (2900, 13.10) (3000, 13.00)
};
\addlegendentry{Local search}

\addplot[color=customOrange, fill=customOrange, fill opacity=0.12, draw=none, forget plot]
    coordinates {
    (2000, 17.00) (2100, 16.00) (2200, 16.00) (2300, 15.00) (2400, 15.00)
    (2500, 15.00) (2600, 14.00) (2700, 14.00) (2800, 13.00) (2900, 13.00) (3000, 13.00)
    (3000, 13.00) (2900, 14.00) (2800, 14.00) (2700, 14.00) (2600, 15.00)
    (2500, 15.00) (2400, 15.00) (2300, 16.00) (2200, 16.00) (2100, 17.00) (2000, 18.00)
} \closedcycle;

\addplot[color=black, line width=1pt, dashed, mark=triangle*, mark size=1.5pt, mark options={solid}] coordinates {
    (2000, 11.00) (2100, 11.00) (2200, 11.00) (2300, 11.00) (2400, 10.00) (2500, 10.00) (2600, 10.00) (2700, 9.00) (2800, 9.00) (2900, 8.00) (3000, 8.00)
};
\addlegendentry{Lower bound}

\end{axis}
\end{tikzpicture}
\caption{Diameter vs.\ ISL range (snapshot scenario, $N = 1{,}584$). %
Plotting conventions as in Fig.~\ref{fig:avg_hops_noviability}.}
\label{fig:worst_case_noviability}
\end{figure}

\begin{figure} %
\centering
\begin{tikzpicture}
\begin{axis}[
    width=\columnwidth,
    height=6.2cm,
    xlabel={Max ISL Distance (km)},
    ylabel={Hops},
    xmin=1950, xmax=3050,
    ymin=9.5, ymax=22.5,
    xtick distance=200,
    ytick distance=1,
    grid=major,
    grid style={dotted, gray!30},
    legend style={at={(0.98,0.98)}, anchor=north east, font=\scriptsize, row sep=1pt},
    legend cell align={left}
]

\addplot[color=customCyan, line width=1.2pt, mark=*, mark size=1.5pt] coordinates {
    (2000, 18.28) (2100, 17.38) (2200, 16.98) (2300, 16.02) (2400, 15.70)
    (2500, 15.00) (2600, 14.16) (2700, 14.00) (2800, 13.32) (2900, 13.00) (3000, 12.04)
};
\addlegendentry{Spectral}

\addplot[color=customCyan, fill=customCyan, fill opacity=0.12, draw=none, forget plot]
    coordinates {
    (2000, 18.00) (2100, 17.00) (2200, 16.00) (2300, 16.00) (2400, 15.00)
    (2500, 15.00) (2600, 14.00) (2700, 14.00) (2800, 13.00) (2900, 13.00) (3000, 12.00)
    (3000, 13.00) (2900, 13.00) (2800, 14.00) (2700, 14.00) (2600, 15.00)
    (2500, 15.00) (2400, 16.00) (2300, 17.00) (2200, 17.00) (2100, 18.00) (2000, 19.00)
} \closedcycle;

\addplot[color=customOrange, line width=1.2pt, mark=square*, mark size=1.5pt] coordinates {
    (2000, 20.44) (2100, 19.18) (2200, 18.34) (2300, 18.00) (2400, 17.24)
    (2500, 17.00) (2600, 16.56) (2700, 16.02) (2800, 15.98) (2900, 15.42) (3000, 15.02)
};
\addlegendentry{Local search}

\addplot[color=customOrange, fill=customOrange, fill opacity=0.12, draw=none, forget plot]
    coordinates {
    (2000, 20.00) (2100, 19.00) (2200, 18.00) (2300, 18.00) (2400, 17.00)
    (2500, 17.00) (2600, 16.00) (2700, 16.00) (2800, 15.00) (2900, 15.00) (3000, 15.00)
    (3000, 16.00) (2900, 16.00) (2800, 16.00) (2700, 17.00) (2600, 17.00)
    (2500, 17.00) (2400, 18.00) (2300, 18.00) (2200, 19.00) (2100, 20.00) (2000, 21.00)
} \closedcycle;

\addplot[color=black, line width=1pt, dashed, mark=triangle*, mark size=1.5pt, mark options={solid}] coordinates {
    (2000, 16.00) (2100, 15.00) (2200, 14.00) (2300, 14.00) (2400, 13.00) (2500, 13.00) (2600, 12.00) (2700, 12.00) (2800, 11.00) (2900, 10.00) (3000, 10.00)
};
\addlegendentry{Lower bound}

\end{axis}
\end{tikzpicture}
\caption{Diameter vs.\ ISL range (viability-constrained scenario, $N = 1{,}584$). %
Plotting conventions as in Fig.~\ref{fig:avg_hops_noviability}.}
\label{fig:worst_case_viability}
\end{figure}

\subsection{Runtime Analysis}

Tables~\ref{tab:runtime-total} and~\ref{tab:runtime-breakdown} report runtimes over 50 sequential runs on an %
Intel Core i7-1165G7 CPU using one Julia thread and one BLAS thread. One unreported warm-up run excluded Julia compilation overhead. The runtime benchmarks use $d_{\max}=3{,}500$~km for $N=500$ and $d_{\max}=2{,}500$~km for $N=1{,}584$.

Table~\ref{tab:runtime-total} reports end-to-end runtimes, and Table~\ref{tab:runtime-breakdown} breaks down the optimization time of the proposed framework into the warm start (WS), which runs the local-search heuristic, Stage~1, and Stage~2.

\begin{table}[H]
\centering
\caption{End-to-end runtime in seconds (mean $\pm$ standard deviation).}
\label{tab:runtime-total}
\normalsize
\renewcommand{\arraystretch}{1.15}
\setlength{\tabcolsep}{2.5pt}
\begin{tabular}{c l c c c}
\hline
$N$ & Scenario & Local search & Proposed & Ratio \\
\hline
500 & Snapshot & $31.83\pm1.65$ & $76.66\pm3.76$ & $2.41{\times}$ \\
500 & Viability & $77.46\pm3.34$ & $108.94\pm4.43$ & $1.41{\times}$ \\
1,584 & Snapshot & $131.52\pm7.78$ & $496.91\pm71.36$ & $3.78{\times}$ \\
1,584 & Viability & $243.57\pm13.22$ & $512.65\pm59.59$ & $2.10{\times}$ \\
\hline
\end{tabular}
\end{table}

\begin{table}[H]
\centering
\caption{Proposed-framework runtime breakdown in seconds (mean $\pm$ standard deviation).}
\label{tab:runtime-breakdown}
\normalsize
\renewcommand{\arraystretch}{1.15}
\setlength{\tabcolsep}{4pt}
\begin{tabular}{c l c c}
\hline
$N$ & Component & Snapshot & Viability \\
\hline
500 & WS & $31.65\pm1.66$ & $77.24\pm3.33$ \\
500 & Stage~1 & $36.43\pm1.85$ & $24.17\pm0.98$ \\
500 & Stage~2 & $8.42\pm1.55$ & $7.36\pm0.22$ \\
\hline
1,584 & WS & $127.45\pm7.73$ & $239.12\pm13.21$ \\
1,584 & Stage~1 & $232.61\pm14.77$ & $119.76\pm5.91$ \\
1,584 & Stage~2 & $132.83\pm64.43$ & $149.11\pm54.77$ \\
\hline
\end{tabular}
\end{table}

\subsection{Comparison with Alternative ISL Topologies}
\label{subsec:literature_comparison}

We evaluate a degree-four $+$Grid topology~\cite{bhattacherjee2019network}, with two in-plane links and one link to each adjacent plane. For our randomized-phase constellations, we choose, for each adjacent plane pair, the feasible cyclic matching with the smallest total link length at $t=0$.

We also evaluate the common-offset construction of Rao and Modiano~\cite{rao2025minimum}, in which each satellite keeps two in-plane links and satellite $i$ in one plane links to satellite $i + k$ (modulo $N_s$) in the next, with the same $k$ for every plane pair. Their exhaustive search over $k$ uses a toroidal lattice and ignores constraint~\eqref{eq:distance_constraint}; for our dimensions, it gives diameter/average maximum hops of 17/17 ($N = 500$) and 42/42 ($N = 1{,}584$). With independent plane phases, however, no common $k$ satisfies~\eqref{eq:distance_constraint} at the tested ranges.

Our feasible Rao-inspired variant starts from $+$Grid, allows a different $k$ for each adjacent plane pair, and locally reduces average shortest-path length.

\begin{table}[H]
\centering
\caption{Feasible topology comparison: diameter / average maximum hops.}
\label{tab:literature_comparison}
\footnotesize
\setlength{\tabcolsep}{3pt}
\renewcommand{\arraystretch}{1.15}
\begin{tabular}{lccc}
\toprule
\textbf{Scenario} & \textbf{Spectral} & \textbf{$+$Grid} & \shortstack{\textbf{Rao-inspired}\\\textbf{(ours)}} \\
\midrule
$N=500$, snapshot & \textbf{9.00 / 8.10} & 19 / 19.00 & 17 / 17.00 \\
$N=500$, viability & \textbf{11.00 / 10.64} & 19 / 19.00 & 17 / 17.00 \\
$N=1{,}584$, snapshot & \textbf{12.02 / 11.45} & 42 / 42.00 & 41 / 41.00 \\
$N=1{,}584$, viability & \textbf{15.00 / 13.98} & 42 / 42.00 & 41 / 41.00 \\
\bottomrule
\end{tabular}
\end{table}

Table~\ref{tab:literature_comparison} compares feasible topologies at $d_{\max}=3{,}500$~km ($N=500$) and $2{,}500$~km ($N=1{,}584$). Spectral entries are 50-run means; baseline entries are deterministic. The two feasible baselines give the same hop metrics in both feasibility models. Across the 2,000--3,000~km large-scale sweep, $+$Grid remains at 42/42.00 and our Rao-inspired variant at 41/41.00 (diameter/average maximum hops). %
Both patterned baselines are limited to in-plane and adjacent-plane links, whereas our methods can select any feasible link under the same degree budget.

\subsection{Discussion}
\label{subsec:discussion}

The simulation results yield several key observations.

\paragraph{Spectral advantage}
The spectral framework achieves a lower diameter than the local-search heuristic in every configuration tested, for example, 9 versus 11 hops (about 18\% fewer) at $N = 500$ in the snapshot scenario, and mean diameters of 12.02 versus 15.00 (snapshot) and 15.00 versus 17.00 (viability) at $N = 1{,}584$ with $d_{\max} = 2{,}500$~km. We attribute this to its optimization of a global property, the algebraic connectivity, rather than local perturbations.

\paragraph{Stability--latency trade-off}
The results illustrate the tension between latency performance and operational stability. The snapshot model yields lower-diameter networks by exploiting transient, short-term links. However, this comes at the cost of high link churn, which would necessitate frequent topology updates and complex dynamic routing protocols, potentially negating the latency benefits. In contrast, the viability-constrained model produces %
higher-diameter topologies whose selected links remain geometrically feasible over the modeled orbital period. Such static or semi-static topologies are well-suited for centrally managed mega-constellations where routing tables can be pre-computed and link reconfiguration is minimized. Here, link stability refers only to geometric feasibility over the modeled interval and does not imply resilience to random link or satellite failures or traffic-induced congestion.

\paragraph{Effect of ISL range}
For the spectral and local-search methods, increasing $d_{\max}$ improves diameter %
under both scenarios. Each 100~km increase in $d_{\max}$ yields %
an average diameter reduction of approximately 0.4 hops in the snapshot scenario and 0.5--0.6 hops under viability constraints. This suggests that improvements in laser terminal technology (enabling longer-range links) would directly translate to latency improvements. The gains are more pronounced in the viability-constrained case, where longer-range links are more likely to remain stable across orbital periods.

\section{Conclusion}
\label{sec:conclusion}

We have %
presented a computationally efficient
framework for designing sparse, low-diameter ISL topologies in large-scale LEO constellations.
Simulations on Starlink-like Walker--Delta constellations show that the spectral framework consistently reduces network diameter by %
2.0--3.0 hops in the snapshot scenario and by 1.4--3.0 hops under viability constraints relative to the local-search heuristic. %
Future directions include incorporating distance-weighted path costs to directly target propagation delay, explicit resilience evaluation under link and satellite failures and traffic-induced congestion, adaptive topologies responding to dynamic traffic, and multi-shell architectures with relay satellites for further diameter reduction.

\appendices
\section{Proof of \texorpdfstring{Lemma~\ref{lm:supergradients}}{Lemma 1}}
\label{app:supergradients}

\begin{proof}
(i) Let $\lambda_2$ be simple, so that $\lambda_2 > \lambda_1 = 0$. By matrix perturbation theory~\cite{hornjohnson2012matrix, stewart1990matrix}, the eigenvalues of a symmetric matrix depend continuously on its entries, and a simple eigenvalue is moreover differentiable, with
\begin{align}
    \frac{\partial \lambda_2(\bm{L}(\bm{z}))}{\partial z_e}\bigg|_{\bm{z} = \bm{x}} = (\bm{v}^{(2)})^T \, \frac{\partial \bm{L}(\bm{z})}{\partial z_e} \, \bm{v}^{(2)}.
    \label{eq:simple_eig_derivative}
\end{align}
By continuity, for all $\bm{z}$ near $\bm{x}$, the eigenvalue $0$ of $\bm{L}(\bm{z})$, with eigenvector $\bm{1}$, remains its smallest eigenvalue, so the Courant--Fischer theorem gives $f(\bm{z}) = \lambda_2(\bm{L}(\bm{z}))$, as in the proof of Proposition~\ref{prop:convexity}. Moreover,~\eqref{eq:laplacian} gives $\partial \bm{L}(\bm{z}) / \partial z_e = \bm{b}_e \bm{b}_e^T$. Hence, for every $e = \{i,j\} \in \Epotential$, $f$ is differentiable at $\bm{x}$ with
\begin{align}
    \frac{\partial f(\bm{x})}{\partial x_e} 
    = q_e(\bm{v}^{(2)}) \label{eq:grad_f}
\end{align}
by~\eqref{eq:simple_eig_derivative} and~\eqref{eq:qev}, since $(\bm{v}^{(2)})^T \bm{b}_e \bm{b}_e^T \bm{v}^{(2)} = (\bm{b}_e^T \bm{v}^{(2)})^2$.
That is, $\nabla f(\bm{x}) = \bm{q}(\bm{v}^{(2)})$.

(ii) Let $\tilde{\bm{v}}$ be a unit vector orthogonal to $\bm{1}$. Since $\tilde{\bm{v}} \in \mathcal{S}$,~\eqref{eq:f},~\eqref{eq:quadratic_form}, and~\eqref{eq:f=lambda2} give, for all $\bm{z} \in \mathbb{R}^M$,
\begin{align}
    f(\bm{z}) &\leq \bm{q}(\tilde{\bm{v}})^T \bm{z} \nonumber \\
    &= f(\bm{x}) + \bm{q}(\tilde{\bm{v}})^T (\bm{z} - \bm{x}) + \bigl(\tilde{\bm{v}}^T \bm{L}(\bm{x})\, \tilde{\bm{v}} - \lambda_2\bigr).
    \label{eq:approx_supergrad}
\end{align}
Setting $\bm{z} = \bm{x}$ gives $\tilde{\bm{v}}^T \bm{L}(\bm{x})\, \tilde{\bm{v}} \geq \lambda_2$. Averaging~\eqref{eq:approx_supergrad} over $\tilde{\bm{v}} = \tilde{\bm{v}}^{(k)}$, $k \in \mathcal{K}$, gives $f(\bm{z}) \leq f(\bm{x}) + \bm{d}(\bm{x})^T (\bm{z} - \bm{x}) + \tilde{\varepsilon}$ for all $\bm{z} \in \mathbb{R}^M$, with $\tilde{\varepsilon} \geq 0$ because $\tilde{\lambda}_k \geq \lambda_2$ for all $k$. By~\eqref{eq:k-cluster}, $\tilde{\lambda}_k < \tilde{\lambda}_2 + \varepsilon$ for all $k \in \mathcal{K}$, and averaging gives $\tilde{\varepsilon} < \varepsilon + (\tilde{\lambda}_2 - \lambda_2)$. If $\tilde{\bm{v}}^{(k)} = \bm{v}^{(k)}$ for all $k$, then $\tilde{\lambda}_k = \lambda_k$, so $\tilde{\lambda}_2 = \lambda_2$ and $\tilde{\varepsilon} < \varepsilon$.
\end{proof}

\section*{Acknowledgments}

The authors thank Dr.\ Aravindan Vijayaraghavan for suggesting the spectral graph-theoretic approach.

\bibliographystyle{IEEEtran}
\bibliography{refs-journal}

@book{GareyJohnson1979,
  author    = {Garey, Michael R. and Johnson, David S.},
  title     = {Computers and Intractability: A Guide to the Theory of NP-Completeness},
  publisher = {W. H. Freeman},
  year      = {1979},
  address   = {San Francisco, CA, USA},
}

@article{fan1949theorem,
  title={On a theorem of Weyl concerning eigenvalues of linear transformations I},
  author={Fan, Ky},
  journal={Proceedings of the National Academy of Sciences},
  volume={35},
  number={11},
  pages={652--655},
  year={1949},
  publisher={National Acad Sciences}
}

@book{boyd2004convex,
  title={Convex Optimization},
  author={Boyd, Stephen and Vandenberghe, Lieven},
  year={2004},
  publisher={Cambridge University Press}
}

@article{vandenberghe1996semidefinite,
  title={Semidefinite programming},
  author={Vandenberghe, Lieven and Boyd, Stephen},
  journal={SIAM Review},
  volume={38},
  number={1},
  pages={49--95},
  year={1996},
  publisher={SIAM}
}

@book{bertsekas1999nonlinear,
  title={Nonlinear Programming},
  author={Bertsekas, Dimitri P},
  year={1999},
  publisher={Athena Scientific},
  address={Belmont, MA}
}

@book{nocedal1999numerical,
  title={Numerical Optimization},
  author={Nocedal, Jorge and Wright, Stephen J},
  year={1999},
  publisher={Springer}
}

@article{lewis1996convex,
  title={Convex analysis on the Hermitian matrices},
  author={Lewis, Adrian S},
  journal={SIAM Journal on Optimization},
  volume={6},
  number={1},
  pages={164--177},
  year={1996},
  publisher={SIAM}
}

@article{overton1992sum,
  title={On the sum of the largest eigenvalues of a symmetric matrix},
  author={Overton, Michael L and Womersley, Robert S},
  journal={SIAM Journal on Matrix Analysis and Applications},
  volume={13},
  number={1},
  pages={41--45},
  year={1992},
  publisher={SIAM}
}

@book{stewart1990matrix,
  title={Matrix Perturbation Theory},
  author={Stewart, Gilbert W and Sun, Ji-guang},
  year={1990},
  publisher={Academic Press}
}

@book{hornjohnson2012matrix,
  author={Horn, Roger A. and Johnson, Charles R.},
  title={Matrix Analysis},
  edition={2},
  publisher={Cambridge University Press},
  year={2012},
  doi={10.1017/CBO9781139020411}
}

@article{polyak1964some,
  title={Some methods of speeding up the convergence of iteration methods},
  author={Polyak, Boris T},
  journal={USSR Computational Mathematics and Mathematical Physics},
  volume={4},
  number={5},
  pages={1--17},
  year={1964},
  publisher={Elsevier}
}

@inproceedings{ghosh2006growing,
  title={Growing well-connected graphs},
  author={Ghosh, Arpita and Boyd, Stephen},
  booktitle={Proceedings of the 45th IEEE Conference on Decision and Control},
  pages={6605--6611},
  year={2006},
  organization={IEEE}
}

@book{chung1997spectral,
  title={Spectral Graph Theory},
  author={Chung, Fan R. K.},
  volume={92},
  year={1997},
  publisher={American Mathematical Society}
}

@article{chung1989diameters,
  title={Diameters and eigenvalues},
  author={Chung, Fan R. K.},
  journal={Journal of the American Mathematical Society},
  volume={2},
  number={2},
  pages={187--196},
  year={1989}
}

@article{alon1985lambda,
  title={{$\lambda_1$}, isoperimetric inequalities for graphs, and superconcentrators},
  author={Alon, Noga and Milman, Vitali D.},
  journal={Journal of Combinatorial Theory, Series B},
  volume={38},
  number={1},
  pages={73--88},
  year={1985},
  publisher={Elsevier},
  doi={10.1016/0095-8956(85)90092-9}
}

@misc{fcc21,
  author       = {{Federal Communications Commission}},
  title        = {{Request for Modification of the Authorization for the SpaceX NGSO Satellite System}},
  howpublished = {Report and Order, FCC 21-48A1},
  year         = {2021},
  url          = {https://docs.fcc.gov/public/attachments/FCC-21-48A1.pdf}
}

@inproceedings{bozkurt2017why,
  title={Why is the internet so slow?!},
  author={Bozkurt, Ilker Nadi and Aguirre, Anthony and Chandrasekaran, Balakrishnan and Godfrey, P Brighten and Laughlin, Gregory and Maggs, Bruce and Singla, Ankit},
  booktitle={International Conference on Passive and Active Network Measurement},
  pages={173--187},
  year={2017},
  organization={Springer}
}

@article{chu2018time,
  title={Time division inter-satellite link topology generation problem: Modeling and solution},
  author={Chu, Xiaogeng and Chen, Yuning},
  journal={International Journal of Satellite Communications and Networking},
  volume={36},
  number={2},
  pages={194--206},
  year={2018},
  publisher={Wiley Online Library}
}

@article{cheng2007core,
  title={Core-based shared tree multicast routing algorithms for {LEO} satellite {IP} networks},
  author={Cheng, Lianzhen and Zhang, Jun and Liu, Kai},
  journal={Chinese Journal of Aeronautics},
  volume={20},
  number={4},
  pages={353--361},
  year={2007},
  publisher={Elsevier}
}

@article{wang2022stochastic,
  title={Stochastic geometry-based low latency routing in massive {LEO} satellite networks},
  author={Wang, Ruibo and Kishk, Mustafa A and Alouini, Mohamed-Slim},
  journal={IEEE Transactions on Aerospace and Electronic Systems},
  volume={58},
  number={5},
  pages={3881--3894},
  year={2022},
  publisher={IEEE}
}

@article{zheng2010routing,
  title={A routing strategy with link disruption tolerance for multilayered satellite networks},
  author={Zheng, Gang and Guo, Yanxin},
  journal={International Journal of Communications, Network and System Sciences},
  volume={3},
  number={11},
  pages={835--842},
  year={2010},
  publisher={Scientific Research Publishing}
}

@inproceedings{dai2021distributed,
  title={A distributed congestion control routing protocol based on traffic classification in {LEO} satellite networks},
  author={Dai, Shiyue and Rui, LanLan and Chen, Shiyou and Qiu, Xuesong},
  booktitle={2021 IFIP/IEEE International Symposium on Integrated Network Management (IM)},
  pages={523--529},
  year={2021},
  organization={IEEE}
}

@inproceedings{bhattacharjee2023laser,
  title={Laser inter-satellite link setup delay: Quantification, impact, and tolerable value},
  author={Bhattacharjee, Dhiraj and Chaudhry, Aizaz U and Yanikomeroglu, Halim and Hu, Peng and Lamontagne, Guillaume},
  booktitle={2023 IEEE Wireless Communications and Networking Conference (WCNC)},
  pages={1--6},
  year={2023},
  organization={IEEE}
}

@incollection{kozen1992depth,
  title={Depth-first and breadth-first search},
  author={Kozen, Dexter C},
  booktitle={The Design and Analysis of Algorithms},
  pages={19--24},
  year={1992},
  publisher={Springer}
}

@article{kermani1979virtual,
  title={Virtual cut-through: A new computer communication switching technique},
  author={Kermani, Parviz and Kleinrock, Leonard},
  journal={Computer Networks (1976)},
  volume={3},
  number={4},
  pages={267--286},
  year={1979},
  publisher={Elsevier}
}

@techreport{rfc3031,
  author      = {Rosen, E. and Viswanathan, A. and Callon, R.},
  title       = {Multiprotocol Label Switching Architecture},
  institution = {Internet Engineering Task Force (IETF)},
  year        = {2001},
  month       = {Jan},
  type        = {RFC},
  number      = {3031},
  doi         = {10.17487/RFC3031}
}

@inproceedings{handley2018delay,
  title={Delay is not an option: Low latency routing in space},
  author={Handley, Mark},
  booktitle={Proceedings of the 17th ACM Workshop on Hot Topics in Networks},
  pages={85--91},
  year={2018}
}

@inproceedings{heine2023end,
  title={End-to-end latency evaluation of {LEO} satellite {IoT} systems},
  author={Heine, Simon and Hofmann, Christian and Knopp, Andreas},
  booktitle={IET Conference Proceedings CP873},
  volume={2023},
  number={48},
  pages={216--221},
  year={2023},
  organization={IET}
}

@inproceedings{wang1993structural,
  title={Structural properties of a low {Earth} orbit satellite constellation---the {Walker Delta} network},
  author={Wang, C-J},
  booktitle={Proceedings of MILCOM'93---IEEE Military Communications Conference},
  volume={3},
  pages={968--972},
  year={1993},
  organization={IEEE}
}

@article{xiaogang2016survey,
  title={A survey of routing techniques for satellite networks},
  author={Qi, Xiaogang and Ma, Jiulong and Wu, Dan and Liu, Lifang and Hu, Shaolin},
  journal={Journal of Communications and Information Networks},
  volume={1},
  number={4},
  pages={66--85},
  year={2016},
  publisher={PTP}
}

@article{sun2022inter,
  title={Inter-satellite time synchronization and ranging link assignment for autonomous navigation satellite constellations},
  author={Sun, Leyuan and Yang, Jun and Huang, Wende and Xu, Laping and Cao, Shaochuan and Shao, Haidong},
  journal={Advances in Space Research},
  volume={69},
  number={6},
  pages={2421--2432},
  year={2022},
  publisher={Elsevier}
}

@article{liu2022data,
  title={A data-driven parallel adaptive large neighborhood search algorithm for a large-scale inter-satellite link scheduling problem},
  author={Liu, Jinming and Xing, Lining and Wang, Ling and Du, Yonghao and Yan, Jungang and Chen, Yingguo},
  journal={Swarm and Evolutionary Computation},
  volume={74},
  pages={101124},
  year={2022},
  publisher={Elsevier}
}

@inproceedings{corici2024sdn,
  title={An {SDN}-based solution for mega-constellation routing},
  author={Corici, Marius and Buhr, Hauke and Zope, Hemant and Zaboub, Manar},
  booktitle={2024 IEEE 35th International Symposium on Personal, Indoor and Mobile Radio Communications (PIMRC)},
  pages={1--6},
  year={2024},
  organization={IEEE}
}

@article{chen2021analysis,
  title={Analysis of inter-satellite link paths for {LEO} mega-constellation networks},
  author={Chen, Quan and Giambene, Giovanni and Yang, Lei and Fan, Chengguang and Chen, Xiaoqian},
  journal={IEEE Transactions on Vehicular Technology},
  volume={70},
  number={3},
  pages={2743--2755},
  year={2021},
  publisher={IEEE}
}

@article{chen2024shortest,
  title={Shortest path in {LEO} satellite constellation networks: An explicit analytic approach},
  author={Chen, Quan and Yang, Lei and Zhao, Yong and Wang, Yi and Zhou, Haibo and Chen, Xiaoqian},
  journal={IEEE Journal on Selected Areas in Communications},
  volume={42},
  number={5},
  pages={1175--1187},
  year={2024},
  publisher={IEEE}
}

@article{berry2024spectrum,
  title={Spectrum Rights in Outer Space: Interference Management for Low Earth Orbit ({LEO}) Broadband Constellations},
  author={Berry, Randall and Bustamante, Pedro and Guo, Dongning and Hazlett, Thomas and Honig, Michael and Murtazashvili, Ilia and Palo, Scott and Weiss, Martin\_B H},
  journal={Journal of Information Policy},
  volume={14},
  year={2024},
  publisher={Pennsylvania State University Press}
}

@inproceedings{mollakhani2025fault,
  title={Fault-Tolerant Spectrum Usage Consensus for Low-Earth-Orbit Satellite Constellations},
  author={Mollakhani, Arman and Guo, Dongning},
  booktitle={2025 IEEE International Conference on Decentralized Applications and Infrastructures (DAPPS)},
  pages={21--26},
  year={2025},
  organization={IEEE}
}

@inproceedings{mollakhani2026inter,
  author    = {Mollakhani, Arman and Tiamraj, Jerayu and Cao, Shu-Jie and Guo, Dongning},
  title     = {Inter-{S}atellite {L}ink {C}onfiguration for {F}ast {D}elivery in {L}ow-{E}arth-{O}rbit {C}onstellations},
  booktitle = {Proceedings of the 2026 IEEE Aerospace Conference (AeroConf)},
  year      = {2026},
  month     = {Mar.},
  address   = {Big Sky, MT, USA},
  organization = {IEEE}
}

@article{hu2025station,
  title={Station Maintenance for Low-Orbit Large-Scale Constellations Based on Absolute and Relative Control Strategies},
  author={Hu, Min and Li, Feifei and Xue, Wen and Liu, Chenhu and Guo, Wen and Ruan, Yongjing},
  journal={Applied Sciences},
  volume={15},
  number={9},
  pages={4640},
  year={2025},
  publisher={MDPI}
}

@inproceedings{bhattacherjee2019network,
  author={Bhattacherjee, Debopam and Singla, Ankit},
  title={Network Topology Design at 27,000 km/hour},
  booktitle={Proceedings of the 15th International Conference on Emerging Networking Experiments and Technologies},
  pages={341--354},
  year={2019},
  doi={10.1145/3359989.3365407}
}

@inproceedings{rao2025minimum,
  author={Rao, Chirag and Modiano, Eytan},
  title={Minimum-hop Constellation Design for Low Earth Orbit Satellite Networks},
  booktitle={IEEE INFOCOM 2025 - IEEE Conference on Computer Communications},
  pages={1--10},
  year={2025},
  doi={10.1109/INFOCOM55648.2025.11044684}
}

@article{metropolis1953equation,
  author  = {Metropolis, Nicholas and Rosenbluth, Arianna W. and Rosenbluth, Marshall N. and Teller, Augusta H. and Teller, Edward},
  title   = {Equation of State Calculations by Fast Computing Machines},
  journal = {The Journal of Chemical Physics},
  volume  = {21},
  number  = {6},
  pages   = {1087--1092},
  year    = {1953},
  doi     = {10.1063/1.1699114}
}

@article{hastings1970monte,
  author  = {Hastings, W. Keith},
  title   = {Monte {C}arlo Sampling Methods Using {M}arkov Chains and Their Applications},
  journal = {Biometrika},
  volume  = {57},
  number  = {1},
  pages   = {97--109},
  year    = {1970},
  doi     = {10.1093/biomet/57.1.97}
}

@book{schrijver2003combinatorial,
  author    = {Schrijver, Alexander},
  title     = {Combinatorial Optimization: Polyhedra and Efficiency},
  publisher = {Springer},
  address   = {Berlin, Germany},
  year      = {2003}
}

@book{lehoucq1998arpack,
  author    = {Lehoucq, Richard B. and Sorensen, Danny C. and Yang, Chao},
  title     = {{ARPACK} Users' Guide: Solution of Large-Scale Eigenvalue Problems with Implicitly Restarted {A}rnoldi Methods},
  publisher = {SIAM},
  address   = {Philadelphia, PA, USA},
  year      = {1998}
}

@article{huangfu2018highs,
  author  = {Huangfu, Qi and Hall, J. A. Julian},
  title   = {Parallelizing the dual revised simplex method},
  journal = {Mathematical Programming Computation},
  volume  = {10},
  number  = {1},
  pages   = {119--142},
  year    = {2018}
}

@book{rockafellar1970convex,
  author    = {Rockafellar, R. Tyrrell},
  title     = {Convex Analysis},
  publisher = {Princeton University Press},
  address   = {Princeton, NJ, USA},
  year      = {1970}
}

\end{document}